\documentclass[letterpaper, 10 pt, conference]{ieeeconf}
\IEEEoverridecommandlockouts
\usepackage{cite}
\usepackage{amsmath,amssymb,amsfonts}
\allowdisplaybreaks
\usepackage{algorithm}
\usepackage{algorithmic}
\usepackage{graphicx}
\usepackage{textcomp}
\usepackage{booktabs}
\usepackage[table,dvipsnames]{xcolor}
\newenvironment{proofsketch}{\par\noindent\hspace{2em}{\itshape Proof sketch: }}{\hspace*{\fill}~\QED\par}
\newcommand{\runin}[1]{\par\medskip\noindent\textbf{#1}\enspace\ignorespaces}

\newtheorem{theorem}{Theorem}[section]

\newtheorem{proposition}[theorem]{Proposition}

\newtheorem{definition}[theorem]{Definition}
\newtheorem{remark}[theorem]{Remark}

\newtheorem{example}[theorem]{Example}
\newcommand{\R}{\mathbb{R}}
\newcommand{\N}{\mathbb{N}}
\title{\LARGE \bf Integer-Free Control Synthesis for Signal Temporal Logic via\\ Ranked Implicants}
\author{Arash Bahari Kordabad$^{1}$, Mohammad H. Mamduhi$^{2}$, and Sadegh Soudjani$^{1,2}$%
\thanks{$^{1}$Max Planck Institute for Software Systems,
Kaiserslautern, Germany. Corresponding author's email:
\texttt{arashbk@mpi-sws.org}}%
\thanks{$^{2}$School of Computer Science, University of Birmingham,
Birmingham, United Kingdom.}%
}
\begin{document}
\maketitle
\thispagestyle{empty}
\pagestyle{empty}
\begin{abstract}
This paper proposes an integer-free method for control synthesis of nonlinear
discrete-time systems under finite-horizon Signal Temporal Logic (STL) constraints.
Instead of introducing integer variables, we consider the
discrete choices of the specification in a ranked order and solve a
standard nonlinear program for each selected choice until a satisfying
trajectory is found. We first define \emph{implicants} as sets of timed
atomic constraints whose joint satisfaction is sufficient to satisfy
the specification. A reference trajectory obtained from a smooth
approximation of the STL robustness is used to rank these implicants.
Fixing one implicant turns the problem into a standard nonlinear
program with no auxiliary weights or bilinear terms. If the solution
does not satisfy the specification, the next implicant will be considered.
We prove soundness and completeness of our method relative to the nonlinear
programming solver. Numerical evaluation on nonlinear multi-agent systems shows that our
proposed method identifies satisfying trajectories where the exact
continuous reformulation fails, and moreover, it is faster than a
mixed-integer formulation of the same specification.
\end{abstract}

\section{Introduction}
\label{sec:introduction}

Signal Temporal Logic (STL) \cite{maler2004} is a
formal language for specifying spatio-temporal requirements of
dynamical systems. Its quantitative
semantics, a.k.a. the robustness function \cite{donze2010}, assigns a real
value to a trajectory, where a positive value implies satisfaction,
allowing STL requirements to be included in trajectory optimization
problems. The robustness function is built from nested $\min$ and
$\max$ operators and is generally nonsmooth and nonconvex. For linear
systems and predicates, STL constraints can be encoded exactly using
mixed-integer linear programming \cite{raman2014}, and more compact
encodings introduce binary variables only at the disjunctive choices
of the formula \cite{kurtz2022}. For nonlinear dynamics, however, the same
encodings lead to nonlinear mixed-integer problems, which can be
significantly harder to solve. Smooth approximations of the robustness function
\cite{pant2017} allow gradient-based solvers, scaled to
high-dimensional nonlinear systems by differential dynamic programming
\cite{kurtz2021ddp}, but they introduce approximation error. In
all of these gradient-based methods, the discrete choices of the specification stay inside
one nonconvex program. A local solver can therefore settle on
an infeasible choice~\cite{kurtz2021ddp}.

Optimization techniques based on the robustness function have also been extended to stochastic
settings~\cite{kordabad2024dr,kordabad2025ddr}.
Control barrier
function (CBF) techniques can also be used to enforce STL specifications
\cite{kordabad2024cbf}, reducing online computation,
but these methods require the barrier functions to be constructed in advance. More
recently, exact continuous reformulations rewrite the $\min$ and
$\max$ operations using continuous weights on simplices
\cite{han2025exact}, or rewrite arbitrary logical constraints through
their conjunctive normal form into max-min constraints
\cite{wehbeh2026logic}. Both remove the binary variables without
approximation error. However, we show that the simplex-weight reformulation
of~\cite{han2025exact} remains sensitive to initialization.

\emph{Contributions.}
First, we show that the exact simplex-weight reformulation in
\cite{han2025exact} is sensitive to initialization. Although
binary variables are removed, the discrete choices remain encoded in
the continuous simplex weights. We characterize this failure mechanism
and show why several natural initialization strategies do not resolve
it. Second, building on the disjunctive view of STL satisfaction
\cite{ho2022automaton,decarli2026stlgcs,ye2026zeroshot}, we formalize
implicants for arbitrary finite-horizon STL formulas, including nested
operators and \emph{until}. We prove that STL satisfaction is equivalent to
the satisfaction of at least one implicant and express the STL
robustness exactly as the maximum of the margins of these implicants.
This separates the discrete choices of the specification from the
continuous trajectory optimization. Third, we propose ranked implicant synthesis. One smooth
robustness problem provides a reference trajectory for ranking the
implicants, and each selected implicant is enforced through a smooth
nonlinear program (NLP) with no simplex weights or bilinear terms. We
prove soundness and completeness relative to the NLP solver and
demonstrate improved feasibility and speed in multi-agent problems.

The proposed search is related to relax-and-fix methods in integer
programming \cite{wolsey2020} and to Logic-Geometric Programming in
task and motion planning \cite{toussaint2015}. Closer to our
approach, \cite{ho2022automaton} and the graph-of-convex-sets planner for STL~\cite{decarli2026stlgcs}
exploit the disjunctive structure of restricted STL settings to guide
sampling-based and convex optimization, respectively. Closely
related, \cite{ye2026zeroshot} rewrites STL into a strengthened
disjunctive normal form and heuristically selects a single branch for
planning. In contrast, our implicants represent the specification
exactly and are \emph{ranked}, rather than selected. A single smooth
reference trajectory orders them by predicate margin, and the
resulting NLPs are tried in that order until a feasible one is found.
Thus, an inaccurate ranking affects efficiency but cannot discard a
feasible implicant, yielding completeness relative to the NLP solver.

\emph{Notation.} $\R$ and $\N$ denote the real numbers and the
positive integers, and $\N_0=\N\cup\{0\}$. Bold letters denote finite
sequences, $\mathbf x=[x_0,\dots,x_N]$. Calligraphic letters denote
sets, families of sets, and trees, $|C|$ is the cardinality of a set $C$, and
$\emptyset$ is the empty set. $\top$ and $\bot$ denote true and
false.

\emph{Outline.} Section~\ref{sec:preliminaries} introduces the STL
formulation and states the problem. Section~\ref{sec:exact} reviews
the exact continuous reformulation and characterizes its sensitivity
to initialization. Section~\ref{sec:method} introduces the implicants,
the proposed ranking and search method, and its guarantees.
Section~\ref{sec:experiments} presents the numerical evaluation, and
Section~\ref{sec:con} concludes the paper.
\section{Preliminaries and Problem Formulation}
\label{sec:preliminaries}
We consider the deterministic nonlinear discrete-time system
\begin{equation}\label{eq:dyn}
x_{k+1}=f(x_k,u_k),
\end{equation}
where $x_k\in\R^{n_x}$ is the state, $u_k\in\mathcal U\subseteq\R^{n_u}$
is the control input, $k\in\N_0$ is the time index, and
$f:\R^{n_x}\times\R^{n_u}\to\R^{n_x}$ is continuously differentiable.
The initial state $x_0$ is given. Over a finite horizon $N\in\N$, we
write $\mathbf x:=[x_0,\dots,x_N]$ and
$\mathbf u:=[u_0,\dots,u_{N-1}]$. We consider finite-horizon STL formulas generated by
\begin{equation*}
\varphi ::= \top
\mid \pi
\mid \neg\varphi
\mid \varphi\wedge\psi
\mid \varphi\,U_{[a,b]}\,\psi ,
\end{equation*}
where $\varphi$ and $\psi$ are STL formulas,
$\pi=\{\alpha(x)\ge0\}$ is a predicate with a continuously
differentiable function $\alpha:\R^{n_x}\to\R$,
and $0\le a\le b$ are integer time bounds. The symbols $\top$,
$\neg$, $\wedge$, and $U_{[a,b]}$ denote \emph{true}, \emph{negation},
\emph{conjunction}, and the bounded \emph{until} operator, respectively. The commonly used disjunction, eventually, and always operators are
defined by $\varphi\vee\psi
:=\neg(\neg\varphi\wedge\neg\psi)$, $\Diamond_{[a,b]}\varphi
:=\top\,U_{[a,b]}\varphi$ and $\Box_{[a,b]}\varphi
:=\neg\Diamond_{[a,b]}\neg\varphi$. Satisfaction of $\varphi$ by
$\mathbf x$ at time $k$, denoted $(\mathbf x,k)\models\varphi$, is
defined by
\begin{align*}
(\mathbf x,k)\models\top
&\quad\text{always},\\
(\mathbf x,k)\models\pi
&\Leftrightarrow  \alpha(x_k)\ge0,\\
(\mathbf x,k)\models\neg\varphi
&\Leftrightarrow (\mathbf x,k)\not\models\varphi,\\
(\mathbf x,k)\models\varphi\wedge\psi
&\Leftrightarrow  (\mathbf x,k)\models\varphi
\text{ and } (\mathbf x,k)\models\psi,\\
(\mathbf x,k)\models\varphi\vee\psi
&\Leftrightarrow  (\mathbf x,k)\models\varphi
\text{ or } (\mathbf x,k)\models\psi,\\
(\mathbf x,k)\models\Box_{[a,b]}\varphi
&\Leftrightarrow  \forall k'\!\in\!\{a,\dots,b\}\!:
(\mathbf x,k+k')\models\varphi,\\
(\mathbf x,k)\models\Diamond_{[a,b]}\varphi
&\Leftrightarrow  \exists k'\!\in\!\{a,\dots,b\}\!:
(\mathbf x,k+k')\models\varphi,\\
(\mathbf x,k)\models\varphi U_{[a,b]}\psi
&\Leftrightarrow \exists k'\!\in\!\{a,\dots,b\}\!:
(\mathbf x,k+k')\models\psi\\*
&\;\;\text{and }\forall k''\!\in\!\{0,\dots,k'\}\!:
(\mathbf x,k+k'')\models\varphi.
\end{align*}
We assume the horizon of $\varphi$ is at most $N$. Hence, satisfaction at
time $0$ depends only on $x_0,\dots,x_N$. A \emph{timed atom} is an inequality
$\alpha(x_k)\ge0$ for a predicate function $\alpha$ of $\varphi$ and a
time $k\in\{0,\dots,N\}$.

The STL robustness function \cite{donze2010} is defined recursively as
\begin{align*}
\rho^\top(\mathbf x,k)
&=+\infty,\\
\rho^{\pi}(\mathbf x,k)
&=\alpha(x_k),\\\rho^{\neg\varphi}(\mathbf x,k)
&=-\rho^\varphi(\mathbf x,k),\\
\rho^{\varphi\wedge\psi}(\mathbf x,k)
&=\min\big(\rho^{\varphi}(\mathbf x,k),\rho^{\psi}(\mathbf x,k)\big),\\
\rho^{\varphi\vee\psi}(\mathbf x,k)
&=\max\big(\rho^{\varphi}(\mathbf x,k),\rho^{\psi}(\mathbf x,k)\big),\\
\rho^{\Box_{[a,b]}\varphi}(\mathbf x,k)
&=\min_{k'\in\{a,\dots,b\}}\rho^{\varphi}(\mathbf x,k+k'),\\
\rho^{\Diamond_{[a,b]}\varphi}(\mathbf x,k)
&=\max_{k'\in\{a,\dots,b\}}\rho^{\varphi}(\mathbf x,k+k'),\\
\rho^{\varphi U_{[a,b]}\psi}(\mathbf x,k)
&=\max_{k'\in\{a,\dots,b\}}\min\Big(\rho^{\psi}(\mathbf x,k+k'),\\*
&\hphantom{{}=\max_{k'\in\{a,\dots,b\}}}
\min_{k''\in\{0,\dots,k'\}}\rho^{\varphi}(\mathbf x,k+k'')\Big),
\end{align*}
and satisfies $\rho^\varphi(\mathbf x,k)>0\Rightarrow
(\mathbf x,k)\models\varphi$. Since strict inequalities are not
convenient in numerical optimization, we enforce a fixed margin
$\delta>0$ and say that $\mathbf x$ satisfies $\varphi$ \emph{with
margin $\delta$} if $\rho^\varphi(\mathbf x,k)\ge\delta$. This implies
$(\mathbf x,k)\models\varphi$ and is the notion enforced in
\eqref{eq:prob}. The margin also settles the strictness of negated
predicates. The formula $\neg\pi$ requires $\alpha(x_k)<0$, and
$\rho^{\neg\pi}(\mathbf x,k)=-\alpha(x_k)\ge\delta$ gives
$\alpha(x_k)\le-\delta<0$.

\runin{A. Problem statement:}
Given a continuously differentiable cost
$J:\R^{n_x(N+1)}\times\R^{n_uN}\to\R$, we consider
\begin{equation}\label{eq:prob}
\min_{\mathbf{x},\mathbf{u}}\;J(\mathbf{x},\mathbf{u})
\quad\text{s.t.}\quad
\eqref{eq:dyn},\quad
u_k\in\mathcal U,\quad
\rho^\varphi(\mathbf{x},0)\ge\delta.
\end{equation}
The main difficulty comes from the STL robustness constraint, since
$\rho^\varphi$ contains nested $\min$ and $\max$ operators. The $\max$
operators associated with $\vee$, $\Diamond$, and $U$ represent
choices, such as which disjunct is satisfied or at which time a
predicate is satisfied. The $\min$ operators associated with $\wedge$,
$\Box$, and $U$ require all corresponding conditions to hold. Both
make the robustness function nonsmooth. A standard approach represents the choices using integer variables and
solves the resulting mixed-integer problem by
branch-and-bound~\cite{nemhauser1988}, which for nonlinear dynamics
becomes a nonlinear mixed-integer problem~\cite{burer2012non}. The
exact reformulation of Section~\ref{sec:exact} replaces the discrete
choices with continuous simplex weights, and smooth robustness methods
replace the nonsmooth $\min$ and $\max$ operators with smooth
approximations. In this work, we use a smooth problem only to rank the
possible choices, then fix one choice at a time and solve the
resulting NLP.

\section{Exact Continuous Reformulation and Its Hidden Decision}
\label{sec:exact}
This section reviews the exact reformulation of
\cite{han2025exact} and explains its sensitivity to initialization. The recursive definition of $\rho^\varphi$ can be represented by a
\emph{robustness tree} $\mathcal T^\varphi$~\cite{takayama2025stlccp}. Each leaf $v$ is a predicate function
evaluated at a specific time, with value
$\rho_v(\mathbf x)=\alpha_v(x_{k_v})$, and each internal node $v$ is a
$\min$ or $\max$ operator over its children
$v_1,\dots,v_m$, with value
$\rho_v(\mathbf x)=\min_j\rho_{v_j}(\mathbf x)$ or
$\rho_v(\mathbf x)=\max_j\rho_{v_j}(\mathbf x)$, respectively. The
value at the root $r$ is the robustness
$\rho_r(\mathbf x)=\rho^\varphi(\mathbf x,0)$. The exact reformulation introduces one
auxiliary scalar $\hat\rho_v$ per node and, \emph{only at $\max$
nodes}, auxiliary weights on the simplex
$\Delta_m:=\{\lambda\in\R^m_{\ge0}:\sum_j\lambda_j=1\}$, with the
constraints
\begin{subequations}
\begin{align}
\text{leaf }v:\quad & \hat\rho_v=\alpha_v(x_{k_v}),\notag\\*
\text{$\min$ node }v:\quad & \hat\rho_v\le\hat\rho_{v_j},
\quad j=1,\dots,m,
\label{eq:minnode}\\*
\text{$\max$ node }v:\quad &
\hat\rho_v\le\textstyle\sum_{j=1}^{m}\lambda^v_j\,\hat\rho_{v_j},
\qquad \lambda^v\in\Delta_m,
\label{eq:maxnode}\\*
\text{root }r:\quad & \hat\rho_r\ge \delta.\notag
\end{align}    
\end{subequations}
No weights are needed at $\min$ nodes, because requiring a common lower bound
for all children is already exact, which reflects that a conjunction
involves no discrete choice.

\begin{proposition}\label{prop:exact}
A trajectory $\mathbf x$ satisfies $\rho^\varphi(\mathbf x,0)\ge\delta$ if
and only if some $\{\hat\rho_v\}$ and $\{\lambda^v\}$ satisfy
all the constraints above, including the leaf equalities and
$\hat\rho_r\ge\delta$.
\end{proposition}
\begin{proofsketch}
At a $\max$ node, $\hat\rho_v\le\sum_j\lambda^v_j\hat\rho_{v_j}\le
\max_j\hat\rho_{v_j}$ for any $\lambda^v\in\Delta_m$, with equality attained by assigning positive weight only to
maximizing children.
The $\min$ nodes and the induction along the tree are immediate. See
\cite[Lemmas~2--4]{han2025exact} for the full argument.
\end{proofsketch}
The reformulation contains no binary variables and has the same
feasible set as \eqref{eq:prob}, but introduces the bilinear terms
$\lambda^v_j\hat\rho_{v_j}$ in \eqref{eq:maxnode}. As shown in the
proof of Proposition~\ref{prop:exact}, the weights select maximizing
children and therefore represent the discrete choices associated with
the $\max$ nodes. Thus, removing the binary variables does not remove
these choices from the continuous optimization problem.
\runin{A. An illustrative example of the initialization issue:}\label{sec:diagnosis}
\begin{example}\label{ex:gamble}
The sensitivity to initialization can be seen in one dimension. Let $x_1=u$, with
$u\in[-1,1]$, $N=1$, and cost $J=u^2$. Consider
$\varphi=\Diamond_{[1,1]}(\pi_1\vee\pi_2)$, where
$\alpha_1(x)=\tfrac14-(x-2)^2$ and
$\alpha_2(x)=\tfrac14-(x+1)^2$.
The first predicate requires $x_1\in[1.5,2.5]$ and is therefore
unreachable, whereas the second requires $x_1\in[-1.5,-0.5]$ and is
reachable, with optimal solution $u^\star=-\tfrac12$. We set $\delta=0$ for the sake of clarity here.  The exact reformulation is
\begin{equation*}
\min_{u\in[-1,1],\,\lambda\in\Delta_2} u^2
\quad\text{s.t.}\quad
\lambda_1\alpha_1(u)+\lambda_2\alpha_2(u)\ge0 .
\end{equation*}
Although this reformulation is exact, its continuous formulation is
nonconvex and can be sensitive to initialization. Indeed,
$\alpha_1(u)-\alpha_2(u)=6u-3$, and for $u>\tfrac12$ the first
predicate has the larger margin, even though it is unreachable.
Consequently, an initialization in this region can drive the local
solver toward the first simplex vertex and the point $u=1$. At this
point, however, $\alpha_1(1)=-\tfrac34<0$, and the corresponding
branch cannot satisfy the specification. In contrast, the feasible
solution $(u,\lambda)=(-\tfrac12,(0,1))$ exists and is optimal.
\end{example}

Thus, the exact reformulation preserves feasibility and the
global optimal value, but a local numerical method may favor an
incorrect disjunct and fail to reach the feasible one.
\runin{B. Alternative approaches to fix the initialization issue:}\label{sec:repairs}
Example~\ref{ex:gamble} additionally shows why several natural attempts to
improve the initialization of the weights may still fail.

\emph{(a) Slack continuation.}
One may relax the constraint as
$\lambda^\top\alpha(u)\ge-s$ and gradually decrease $s$ to zero.
For a large $s$, there is little incentive for the weights to favor
either goal. At $\lambda=(\tfrac12,\tfrac12)$, the largest averaged
margin is $\max_u \tfrac12(\alpha_1(u)+\alpha_2(u))=-2$, attained at
$u=\tfrac12$. Hence, reducing $s$ below $2$ requires the weights to
move toward one goal. At $u=\tfrac12$, however,
$\alpha_1(u)=\alpha_2(u)=-2$, the local information does not favor
either choice, and a small change toward $u>\tfrac12$ favors the
unreachable first goal. Hence, slack continuation alone does not
resolve the choice between the two goals, and the solver may still
end at the unreachable one.

\emph{(b) Alternating optimization.}
Another possibility is to optimize $\lambda$ and $u$ in turn. For a
fixed $u$, the best weight is placed on the predicate with the larger
margin. Hence, if $u^0>\tfrac12$, the first weight update selects goal
$1$. Optimizing $u$ with this choice moves the solution toward $u=1$,
where $\alpha_1(1)>\alpha_2(1)$, and the next weight update again
selects goal $1$. The iterations remain at the unreachable goal even
though goal $2$ is feasible.

\emph{(c) Warm-starting the weights.}
Initializing the weights at the correct choice does not solve
the problem. For $\lambda^0=(0,1)$ and $u^0=0.9$, we have
$\alpha_1(u^0)-\alpha_2(u^0)=2.4>0$. Since $\lambda$ is a
decision variable, a local solver is free to move the weights toward
goal $1$. A good initialization of $\lambda$ does not ensure that the
same choice is maintained during the optimization.

In all three cases, the choice between the goals remains part of the
NLP and depends on the local trajectory. The next section fixes one
implicant while its NLP is solved.
\section{STL Synthesis via Ranked Implicants}
\label{sec:method}

In this section, we make the discrete choices in the STL specification explicit.
Expanding the temporal operators over the finite horizon gives the
sets of timed atomic constraints that can satisfy the specification,
which we call implicants. A smooth surrogate trajectory ranks the implicants.

\runin{A. Implicants via unrolling:}
We first unroll the formula into a propositional formula, as in
mixed-integer encodings of STL~\cite{raman2014} and bounded model
checking~\cite{biere1999}, and push negations to the predicates.
A negated timed atom $\neg(\alpha(x_k)\ge0)$ is replaced by the timed
atom $-\alpha(x_k)\ge0$. The resulting formula is negation-free and
hence in \emph{positive normal form}~\cite{sadraddini2015}. The
replacement leaves the robustness unchanged and changes Boolean
satisfaction only on the boundary $\alpha(x_k)=0$, excluded by the
margin (Remark~\ref{rem:negation}).
\begin{definition}[Unrolling]\label{def:ground}
For an STL formula $\varphi$ and time $k$, the unrolled formula
$G_k(\varphi)$ is defined recursively as
\begin{align*}
G_k(\top)&=\top,\\
G_k(\neg\varphi)&=\neg G_k(\varphi),\\
G_k(\pi)
&=\bigl(\alpha(x_k)\ge0\bigr),\\
G_k(\varphi\wedge\psi)
&=G_k(\varphi)\wedge G_k(\psi),\\
G_k(\varphi\vee\psi)
&=G_k(\varphi)\vee G_k(\psi),\\
G_k(\Box_{[a,b]}\varphi)
&=\textstyle\bigwedge_{k'\in\{a,\dots,b\}}G_{k+k'}(\varphi),\\
G_k(\Diamond_{[a,b]}\varphi)
&=\textstyle\bigvee_{k'\in\{a,\dots,b\}}G_{k+k'}(\varphi),\\
G_k(\varphi U_{[a,b]}\psi)
&=\textstyle\bigvee_{k'\in\{a,\dots,b\}}
\Big(
G_{k+k'}(\psi)
\wedge\\*
&\hphantom{{}=\textstyle\bigvee_{k'\in\{a,\dots,b\}}\Big(}
\textstyle\bigwedge_{k''\in\{0,\dots,k'\}}G_{k+k''}(\varphi)
\Big).
\end{align*}
Here, $G_k(\varphi)$ is a propositional formula over timed atoms, and
by the STL semantics $(\mathbf{x},k)\models\varphi$ if and only if
$G_k(\varphi)$ is satisfied by $\mathbf{x}$. We then push negations to
the atoms by De~Morgan's laws, replace every $\neg(\alpha(x_k)\ge0)$
by $-\alpha(x_k)\ge0$, and let $G_k(\varphi)$ denote the resulting
negation-free form.
\end{definition}

\begin{definition}[Implicants]\label{def:implicant}
For a timed atom $c$, define $\mathcal W(c)=\{\{c\}\}$, the family whose only implicant is $\{c\}$. Moreover,
$\mathcal W(\top)=\{\emptyset\}$, the family whose only implicant is the empty implicant, which is always
satisfied, and $\mathcal W(\bot)=\emptyset$, no implicant. For two unrolled
formulas $H_1$ and $H_2$, define
\begin{align*}
\mathcal W(H_1\wedge H_2)
&=
\bigl\{C_1\cup C_2:
C_1\in\mathcal W(H_1),\
C_2\in\mathcal W(H_2)\bigr\},\\*
\mathcal W(H_1\vee H_2)
&=
\mathcal W(H_1)\cup\mathcal W(H_2).
\end{align*}
Then $\mathcal W_k(\varphi):=\mathcal W(G_k(\varphi))$.
Thus, an implicant $C\in\mathcal W_k(\varphi)$ is a set of timed atomic
constraints whose joint satisfaction is sufficient to satisfy
$\varphi$ at time $k$. It is an implicant in the sense of propositional
logic~\cite{ferrere2015}, and it corresponds to a disjunctive branch
in~\cite{ye2026zeroshot,decarli2026stlgcs}.
\end{definition}

\begin{remark}[Prime implicants]\label{rem:minimal}
An implicant that strictly contains another is redundant. Indeed,
if $C\subset C'$ are both implicants, any trajectory satisfying all
atoms of $C'$ also satisfies all atoms of $C$, and imposing the atoms
of $C$ gives a larger feasible set and no larger optimal cost. Hence,
only the subset-minimal implicants, the prime implicants, matter.
\end{remark}

\begin{remark}[Negated predicates and the margin]\label{rem:negation}
Replacing $\neg(\alpha(x_k)\ge0)$ by $-\alpha(x_k)\ge0$ turns the
strict inequality $\alpha(x_k)<0$ into a non-strict one, and the two
differ only on the boundary $\alpha(x_k)=0$. This has no effect on the
method. The robustness is unchanged, since $\rho^{\neg\pi}=-\alpha$,
and every atom of an implicant is enforced with margin $\delta>0$ in
\eqref{eq:committed}. Hence, a returned trajectory satisfies
$-\alpha(x_k)\ge\delta$ and $\alpha(x_k)<0$ for every negated
atom of the selected implicant. Unlike~\cite{sadraddini2015}, a negated predicate adds no
secondary signal, only timed atoms.
\end{remark}

\begin{example}
\label{ex:families}
In this example, we consider three predicates $\pi_A$, $\pi_B$ and $\pi_C$ with predicate functions $\alpha_A$, $\alpha_B$ and $\alpha_C$, respectively. Write $a_k$ for the timed atom $\alpha_A(x_k)\ge0$, and similarly
$b_k$ and $c_k$.
\newline
\textnormal{(i)} For $\Diamond_{[1,3]}\pi_A$, unrolling gives
$G_0=a_1\vee a_2\vee a_3$. Hence, the implicants are
$\{a_1\}$, $\{a_2\}$, and $\{a_3\}$, corresponding to satisfying
$\pi_A$ at time $1$, $2$, or $3$.
\newline
\textnormal{(ii)} For
$(\pi_A\vee\pi_B)\wedge\Diamond_{[1,2]}\pi_C$, unrolling gives
$G_0=(a_0\vee b_0)\wedge(c_1\vee c_2)$. An implicant selects one atom
from each disjunction, giving $\{a_0,c_1\}$, $\{a_0,c_2\}$,
$\{b_0,c_1\}$, and $\{b_0,c_2\}$. Thus, an implicant can represent both
a logical choice and a satisfaction time.
\newline
\textnormal{(iii)} For the nested formula
$\Box_{[0,1]}\Diamond_{[0,2]}\pi_A$, unrolling gives
$G_0=(a_0\vee a_1\vee a_2)\wedge(a_1\vee a_2\vee a_3)$. Choosing one
atom from each parenthesis gives $3\times3=9$ selections. Since
implicants are sets, choosing $a_1$ from both parentheses gives
$\{a_1\}$, and $\{a_1,a_2\}$ is removed by Remark~\ref{rem:minimal}
because it contains $\{a_1\}$. Only $\{a_1\}$, $\{a_2\}$, and
$\{a_0,a_3\}$ remain. Indeed, $a_1$ or $a_2$ alone satisfies both
parentheses, and if neither holds, the first parenthesis requires
$a_0$ and the second requires $a_3$.
\newline
\textnormal{(iv)} For $\neg(\pi_A\,U_{[1,2]}\,\pi_B)$, unrolling gives
$G_0=\neg\big[(b_1\wedge a_0\wedge a_1)\vee(b_2\wedge a_0\wedge a_1\wedge a_2)\big]$.
Pushing the negation to the atoms yields
$(\bar b_1\vee\bar a_0\vee\bar a_1)\wedge(\bar b_2\vee\bar a_0\vee\bar a_1\vee\bar a_2)$,
where $\bar a_k$ denotes the atom $-\alpha_A(x_k)\ge0$ and similarly
$\bar b_k$. The prime implicants of this formula are $\{\bar a_0\}$,
$\{\bar a_1\}$, $\{\bar b_1,\bar b_2\}$ and $\{\bar b_1,\bar a_2\}$.
\end{example}

\begin{theorem}[Equivalence]\label{thm:equiv}
For every finite-horizon STL formula $\varphi$ and every trajectory
$\mathbf{x}$,
\begin{equation}\label{eq:maxmin}
\rho^\varphi(\mathbf{x},0)
=
\max_{C\in\mathcal W_0(\varphi)}
\min_{c\in C}\alpha_c(x_{k_c}),
\end{equation}
with $\min\emptyset=+\infty$ and $\max\emptyset=-\infty$.
Consequently, for every $\delta\ge0$, $\rho^\varphi(\mathbf x,0)\ge\delta$
if and only if there exists $C\in\mathcal W_0(\varphi)$ with
$\alpha_c(x_{k_c})\ge\delta$ for all $c\in C$. The trajectories
satisfying $\varphi$ with margin $\delta>0$ are exactly the union over
the implicants of the sets defined by their atoms with margin $\delta$.
\end{theorem}

\begin{proof}
Unrolling preserves robustness, since each rule of
Definition~\ref{def:ground} mirrors the recursive definition of $\rho$
with $\wedge$ and $\vee$ evaluated as $\min$ and $\max$. Pushing
negations to the atoms preserves it, since $-\min(a,b)=\max(-a,-b)$,
and replacing $\neg(\alpha\ge0)$ by $-\alpha\ge0$ preserves it, since
$\rho^{\neg\pi}=-\alpha$. For a negation-free propositional formula
$G$, $\rho(G)=\max_{C\in\mathcal W(G)}\min_{c\in C}\alpha_c$ follows by
induction. For $H_1\wedge H_2$,
$\min(\max_{C_1}\min_{C_1},\max_{C_2}\min_{C_2})=\max_{C_1,C_2}\min_{C_1\cup C_2}$,
and for $H_1\vee H_2$ the maximum runs over the union of the two
families. The remaining statements follow from \eqref{eq:maxmin}.
\end{proof}

The same choices appear implicitly in mixed-integer STL encodings
\cite{raman2014,kurtz2022} and in the simplex weights of
Section~\ref{sec:exact}. The implicant representation separates the discrete choices
from the continuous optimization, without binary variables or
auxiliary simplex weights.

\runin{B. Reference trajectory:}
To determine which implicants should be considered first, we use a
smooth approximation of the STL robustness. We recall the standard
smooth approximations of $\min$ and $\max$~\cite{gilpin2021}. For
$n\in\N$, $z\in\R^n$ and $\beta>0$:
\begin{equation}\label{eq:lse}
\begin{aligned}
\widetilde{\max}_\beta(z)
&:=
\frac{1}{\beta}\log\Big(\sum_{j=1}^{n}e^{\beta z_j}\Big),\\
\widetilde{\min}_\beta(z)
&:=
-\widetilde{\max}_\beta(-z).
\end{aligned}
\end{equation}
As $\beta\to\infty$, these approximations converge to $\max_j z_j$
and $\min_j z_j$, respectively. Let
$\tilde\rho^\varphi_\beta(\mathbf{x})$ denote the smooth approximation
of $\rho^\varphi(\mathbf{x},0)$ obtained by recursively replacing each
$\max$ and $\min$ operator by \eqref{eq:lse}, which converges to
$\rho^\varphi(\mathbf{x},0)$ as $\beta\to\infty$. The reference trajectory $\mathbf{x}^r$ is
obtained by solving
\begin{equation}\label{eq:smoothprob}
\max_{\mathbf{x},\mathbf{u}}\;
\tilde\rho^\varphi_\beta(\mathbf{x})
-\gamma J(\mathbf{x},\mathbf{u})
\quad
\text{s.t.}\quad
\eqref{eq:dyn},\;
u_k\in\mathcal{U},
\end{equation}
where $\gamma>0$ is a constant that weights the two terms. A
small $\gamma$ favors the STL robustness over the cost, and a large $\gamma$ the cost over the robustness function. The trajectory $\mathbf{x}^r$ is used only
to rank the implicants and to initialize the induced programs, and it is neither a tracking reference, nor
required to satisfy the STL
specification. Feasibility of \eqref{eq:smoothprob} is therefore
independent of STL satisfaction. It only requires an admissible input
sequence for which the dynamics are well defined over the finite
horizon.

\runin{C. Ranking:}
For each implicant $C$, we define the score
\begin{equation}\label{eq:score}
S(C):=\min_{c\in C}\alpha_c(x^r_{k_c}),
\end{equation}
which is the smallest predicate margin among the atoms of $C$ along
$\mathbf{x}^r$ and measures how close the implicant is to satisfaction,
since all its atoms must hold. We consider the implicants in
decreasing order of $S(C)$. The score determines only this order and
never removes an implicant. 
\runin{D. Induced program:}
For the selected implicant $C$, we solve the \emph{induced program}
\begin{subequations}\label{eq:committed}
\begin{align}
\mathrm{NLP}(C):\quad
\min_{\mathbf{x},\mathbf{u}}\quad & J(\mathbf{x},\mathbf{u})\\*
\text{s.t.}\quad
& \eqref{eq:dyn},\quad u_k\in\mathcal{U},\\*
& \alpha_c(x_{k_c})\ge\delta,\quad \forall c\in C,
\end{align}
\end{subequations}
for a fixed implicant $C$, where $\delta>0$ is a satisfaction margin. The program directly
enforces all atoms of the selected implicant and contains no auxiliary
weights, bilinear terms, or $\min$/$\max$ operators. By
Theorem~\ref{thm:equiv}, any feasible solution of $\mathrm{NLP}(C)$
satisfies the original STL specification.
If $\mathrm{NLP}(C)$ does not provide a satisfying trajectory, the
next implicant in the ranking is considered. This process continues
until a satisfying trajectory is found or all implicants have been
examined. Identical implicants are considered only once, and
non-prime implicants are removed according to
Remark~\ref{rem:minimal}. Importantly, the ranking changes only the
order of the search. A feasible implicant remains in consideration even
if the reference trajectory assigns it a low score in the ranking.

\begin{figure*}[t!]
\centering
\includegraphics[width=0.9\textwidth]{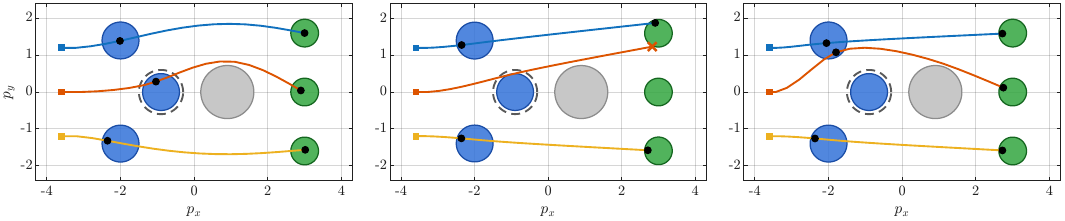}
\caption{The cold-started (left) and
warm-started (middle) exact reformulations fail, while the proposed
method (right) succeeds with $M=3$. Blue disks are the charging
stations, green disks the goals, the gray disk the obstacle, and the
dashed circle the clearance zone.}
\label{fig:traj}
\end{figure*}

\begin{algorithm}[t]
\caption{Ranked implicant search}
\label{alg:wgs}
\begin{algorithmic}[1]
\STATE Solve \eqref{eq:smoothprob} and obtain the reference trajectory
$\mathbf{x}^r$.
\STATE Generate implicants from $\mathcal W_0(\varphi)$ (as needed) in
decreasing order of the score \eqref{eq:score}.
\FOR{each generated implicant $C$ in this order}
\STATE Solve $\mathrm{NLP}(C)$ using $\mathbf{x}^r$ as an initial
    trajectory.
    \IF{$\mathrm{NLP}(C)$ returns a feasible point $(\mathbf{x},\mathbf{u})$ with $(\mathbf{x},0)\models\varphi$}
        \RETURN $(\mathbf{x},\mathbf{u})$.
    \ENDIF
\ENDFOR
\RETURN no satisfying trajectory found.
\end{algorithmic}
\end{algorithm}

\begin{example}[Example~\ref{ex:gamble}, revisited]
\label{ex:implicant}
For $\varphi=\Diamond_{[1,1]}(\pi_1\vee\pi_2)$, unrolling gives
$G_0(\varphi)=(\alpha_1(x_1)\ge0)\vee(\alpha_2(x_1)\ge0)$ and the two
implicants $C_1=\{\alpha_1(x_1)\ge0\}$ and $C_2=\{\alpha_2(x_1)\ge0\}$.
Suppose the surrogate problem \eqref{eq:smoothprob} gives
$x_1^r=0.7$. Then $S(C_1)=\alpha_1(0.7)=-1.44$ and
$S(C_2)=\alpha_2(0.7)=-2.64$, and the method first considers $C_1$.
However, $\mathrm{NLP}(C_1)$ is infeasible, since
$\alpha_1(u)\le-3/4$ for every $u\in[-1,1]$. The method then solves
$\mathrm{NLP}(C_2)$ and obtains $u^\star=-1+\sqrt{\tfrac14-\delta}$,
which approaches $-1/2$ as $\delta\to0$. Although the surrogate ranks
the infeasible implicant first, the search still finds $C_2$. The
surrogate affects only the order of the search.
\end{example}

\begin{theorem}[Guarantees]\label{prop:guar}
(i) \emph{Soundness:} every trajectory returned by
Algorithm~\ref{alg:wgs} is feasible for \eqref{eq:prob}.
(ii) \emph{Completeness relative to the NLP solver:} suppose
\eqref{eq:prob} is feasible and, whenever \eqref{eq:committed} is
feasible, the NLP solver finds a feasible point. Then
Algorithm~\ref{alg:wgs} returns a feasible solution of \eqref{eq:prob}.
\end{theorem}

\begin{proof}
For (i), a feasible point of $\mathrm{NLP}(C)$ satisfies
$\alpha_c(x_{k_c})\ge\delta$ for all $c\in C$, hence
$\rho^\varphi(\mathbf x,0)\ge\delta>0$ by Theorem~\ref{thm:equiv} and
$(\mathbf x,0)\models\varphi$. Negated predicates hold strictly by
Remark~\ref{rem:negation}. For (ii), if \eqref{eq:prob} is feasible, then by
\eqref{eq:maxmin} there exists an implicant $C^\star$ whose atoms all
have margin at least $\delta$. Hence, $\mathrm{NLP}(C^\star)$ is
feasible. Since all implicants are eventually considered, the solver
then solves it.
\end{proof}
\begin{remark}[Scalability]\label{rem:scalability}
The number of implicants affects only the search effort, not the size
of the individual optimization problems. The reference trajectory is
computed once by solving \eqref{eq:smoothprob}, whose size is
independent of $|\mathcal W_0(\varphi)|$, scoring an implicant only
requires evaluating its atomic predicates along $\mathbf{x}^r$, and
each induced program $\mathrm{NLP}(C)$ contains the dynamics, the
input constraints, and the $|C|$ atomic constraints of that implicant.
Thus, a large implicant family may increase the number of candidates
considered, and in the worst case many implicants may still need to be
examined. The ranking places promising implicants early. The full
implicant family need not be constructed explicitly. For specifications
with separate disjunctive choices, the individual choices can be
stored in small ranked lists and combined only as needed.
\end{remark}
\section{Numerical Evaluation}
\label{sec:experiments}

We consider $M$ unicycle agents with positions $p^i_k\in\R^2$ and
headings $\theta^i_k\in\R$, evolving as
\begin{equation*}
p^i_{k+1}=p^i_k+\tau v^i_k
\begin{bmatrix}
\cos\theta^i_k\\
\sin\theta^i_k
\end{bmatrix},
\qquad
\theta^i_{k+1}=\theta^i_k+\tau\omega^i_k,
\end{equation*}
with $\tau=0.5$, $v^i_k\in[0,1]$, $|\omega^i_k|\le1.5$, and
$N=20$. The workspace contains three charging disks, $M$ goal disks,
and a circular obstacle, as shown in Figs.~\ref{fig:traj}
and~\ref{fig:swarm}. The cost is
$J=\sum_{i=1}^{M}\sum_{k=0}^{N-1}
((v^i_k)^2+(\omega^i_k)^2)$.
Each agent must remain outside the clearance zone until it registers
at one of the three charging stations, then reach one of the goals
within twelve steps, while avoiding the obstacle and maintaining
pairwise separation. This can be formally expressed as
\begin{align*}
\varphi=\;&\bigwedge_{i=1}^{M}
\Big[\pi^{i}_{\mathrm{out}}\,U_{[2,8]}
\Big(\bigvee_{w=1}^{3}\pi^{i}_{\mathrm{chg},w}
\wedge\Diamond_{[1,12]}\bigvee_{r=1}^{M}\pi^{i}_{\mathrm{goal},r}\Big)\Big]\\*
&\wedge\bigwedge_{i=1}^{M}\Box_{[0,N]}\pi^{i}_{\mathrm{obs}}
\wedge\!\!\bigwedge_{1\le i<j\le M}\!\!
\Box_{[0,N]}\pi^{ij}_{\mathrm{sep}} .
\end{align*}
All predicates are signed distances. The predicates
$\pi^{i}_{\mathrm{out}}$ and $\pi^{i}_{\mathrm{obs}}$ are positive
outside the clearance zone and the obstacle,
$\pi^{i}_{\mathrm{chg},w}$ and $\pi^{i}_{\mathrm{goal},r}$ inside
charging disk $w$ and goal disk $r$, and $\pi^{ij}_{\mathrm{sep}}$
when agents $i$ and $j$ are at least $d_{\mathrm s}=0.65$ apart. All
NLPs are solved with the solver IPOPT (Interior Point Optimizer)~\cite{wachter2006}, using $\delta=0.02$,
$\beta=8$, and $\gamma=0.05$. The implementation skips implicants
that are infeasible for geometric reasons, such as a charging station
inside the clearance zone or a goal that cannot be reached within its
time bounds, and places implicants in which two agents share a goal
disk behind the others. Neither step removes a feasible implicant. The exact reformulation of
Section~\ref{sec:exact} is tested with a cold start, a straight-line
trajectory with uniform weights, and a warm start, the surrogate
trajectory with every weight vector at the vertex that its margins
select. In the instance of Fig.~\ref{fig:traj}, one charging station lies inside the clearance zone. Since
$\pi^{i}_{\mathrm{out}}$ must hold up to and including the registration
step, no agent can register at this station. Every implicant
that selects this station is therefore infeasible. Both initial
trajectories nevertheless lead the middle agent toward it. As shown in
Fig.~\ref{fig:traj}, both exact reformulations fail for the
middle agent, whereas our proposed method succeeds.

We next increase the number of agents from $M=3$ to $M=8$, in the
benchmark workspace shown in Fig.~\ref{fig:swarm}. For these instances,
we compare the proposed method with a mixed-integer formulation that
introduces binary variables only at the disjunctions.
This nonlinear mixed-integer problem is solved
with the solver SCIP (Solving Constraint Integer
Programs)~\cite{hojny2025scip} until its first satisfying solution,
with a time limit of $10^4$\,s and 10 random seeds per instance. The
proposed method is deterministic, and its computation time is averaged
over three repetitions.
Table~\ref{tab:runtime_cost} shows that SCIP succeeds in every run only
for $M\le4$ and fails in at least 7 of 10 runs for $M\ge6$. With
seed 0, a time limit of $2\times10^4$\,s still yields no satisfying trajectory
for $M=8$. Where SCIP succeeds, its mean computation time is $8$ to $86$
times that of the proposed method, and each of its first satisfying
solutions costs more than the locally optimal solution of the
induced NLP. Fig.~\ref{fig:swarm} shows the proposed solution for
$M=8$.

The ranking is heuristic, and the implicant it places first need
not be feasible. To see how often this happens, the instance with $M=8$
is perturbed 10 times by moving every start and goal position by up to
$0.2$ in each coordinate, with the formula $\varphi$ unchanged.
Table~\ref{tab:ranking} lists, for each instance, the number of induced
NLPs solved until the first satisfying trajectory is found. In 5 of the 10
instances the first induced NLP is infeasible, and the search returns a
satisfying trajectory after at most 6 NLPs and within $223$\,s.

\begin{table}[t]
\centering
\caption{Time and cost. SCIP: satisfying runs of 10 within
$10^4$\,s, and mean (min--max) over them.}
\label{tab:runtime_cost}
\footnotesize
\setlength{\tabcolsep}{3pt}
\begin{tabular}{c rr c rr}
\toprule
& \multicolumn{2}{c}{Proposed}
& \multicolumn{3}{c}{SCIP} \\
\cmidrule(lr){2-3}\cmidrule(lr){4-6}
$M$ & Time [s] & Cost & Feasible & Time [s] & Cost \\
\midrule
3 & 16.9  & 33.2  & 10/10 & 132 (39--314)    & 47.8 (34.3--77.5)  \\
4 & 13.1  & 43.7  & 10/10 & 563 (106--1903)  & 57.9 (50.1--63.4)  \\
5 & 64.6  & 55.2  & 9/10  & 1301 (211--3311) & 72.7 (69.4--80.1)  \\
6 & 42.2  & 68.0  & 3/10  & 2098 (750--3567) & 94.3 (80.5--108.0) \\
7 & 75.8  & 86.0  & 2/10  & 6492 (5970--7014) & 94.9 (94.6--95.3) \\
8 & 134.5 & 112.2 & 1/10  & 5046             & 135.9              \\
\bottomrule
\end{tabular}
\end{table}

\begin{table}[t]
\centering
\caption{10 perturbed instances with $M=8$: induced NLPs solved
until the first satisfying trajectory, and computation time.}
\label{tab:ranking}
\footnotesize
\setlength{\tabcolsep}{3.2pt}
\begin{tabular}{l cccccccccc}
\toprule
Instance & 1 & 2 & 3 & 4 & 5 & 6 & 7 & 8 & 9 & 10 \\
\midrule
NLPs     & 6 & 4 & 1 & 5 & 1 & 3 & 1 & 1 & 1 & 4 \\
Time [s] & 132 & 90 & 66 & 116 & 68 & 101 & 75 & 95 & 157 & 223 \\
\bottomrule
\end{tabular}
\end{table}

\begin{figure}[t!]
\centering
\includegraphics[width=0.7\columnwidth]{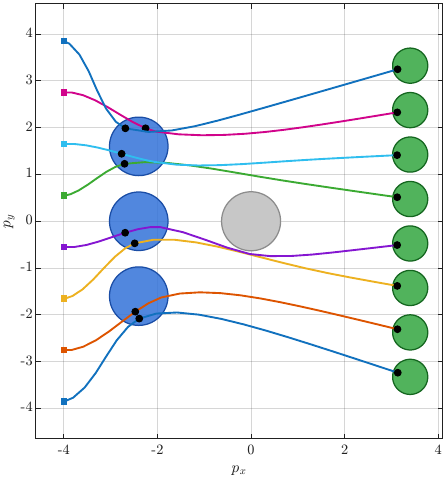}
\caption{Proposed method for the benchmark instance with $M=8$. Colors
as in Fig.~\ref{fig:traj}.}
\label{fig:swarm}
\end{figure}

\section{Conclusion}
\label{sec:con}

We propose an integer-free synthesis method for finite-horizon STL
specifications that makes the
discrete choices of the specification explicit and ranks them using a
smooth reference trajectory. Each selected implicant is solved through
a standard NLP without simplex weights or bilinear terms. The method
is sound and complete relative to the NLP solver, succeeds where the
exact continuous reformulation fails, and scales to instances for
which the mixed-integer formulation reaches its time limit. Possible
extensions include global optimality across implicants, using
cost lower bounds to certify the best solution found, and improved
ranking through additional reference trajectories that explore
alternative disjunctive choices of the specification.

\bibliographystyle{IEEEtran}
\bibliography{integer}

@inproceedings{maler2004,
  author    = {Maler, Oded and Ni{\v{c}}kovi{\'c}, Dejan},
  title     = {Monitoring temporal properties of continuous signals},
  booktitle = {Formal Techniques, Modelling and Analysis of Timed and
               Fault-Tolerant Systems (FORMATS/FTRTFT)},
  series    = {LNCS},
  volume    = {3253},
  pages     = {152--166},
  publisher = {Springer},
  year      = {2004}
}

@inproceedings{donze2010,
  author    = {Donz{\'e}, Alexandre and Maler, Oded},
  title     = {Robust satisfaction of temporal logic over real-valued
               signals},
  booktitle = {Formal Modeling and Analysis of Timed Systems (FORMATS)},
  series    = {LNCS},
  volume    = {6246},
  pages     = {92--106},
  publisher = {Springer},
  year      = {2010}
}

@inproceedings{raman2014,
  author    = {Raman, Vasumathi and Donz{\'e}, Alexandre and Maasoumy,
               Mehdi and Murray, Richard M. and
               Sangiovanni-Vincentelli, Alberto and Seshia, Sanjit A.},
  title     = {Model predictive control with signal temporal logic
               specifications},
  booktitle = {53rd IEEE Conference on Decision and Control (CDC)},
  pages     = {81--87},
  year      = {2014}
}

@article{kurtz2022,
  author  = {Kurtz, Vince and Lin, Hai},
  title   = {Mixed-integer programming for signal temporal logic with
             fewer binary variables},
  journal = {IEEE Control Systems Letters},
  volume  = {6},
  pages   = {2635--2640},
  year    = {2022}
}

@article{kurtz2021ddp,
  author  = {Kurtz, Vince and Lin, Hai},
  title   = {Trajectory optimization for high-dimensional nonlinear
             systems under {STL} specifications},
  journal = {IEEE Control Systems Letters},
  volume  = {5},
  number  = {4},
  pages   = {1429--1434},
  year    = {2021}
}

@inproceedings{ferrere2015,
  author    = {Ferr{\`e}re, Thomas and Maler, Oded and Ni{\v{c}}kovi{\'c}, Dejan},
  title     = {Trace diagnostics using temporal implicants},
  booktitle = {Automated Technology for Verification and Analysis (ATVA)},
  series    = {LNCS},
  volume    = {9364},
  pages     = {241--258},
  year      = {2015}
}

@inproceedings{sadraddini2015,
  author    = {Sadraddini, Sadra and Belta, Calin},
  title     = {Robust temporal logic model predictive control},
  booktitle = {53rd Annual Allerton Conference on Communication, Control,
               and Computing (Allerton)},
  pages     = {772--779},
  year      = {2015}
}

@inproceedings{ho2022automaton,
  author    = {Ho, Qi Heng and Ilyes, Roland B. and Sunberg, Zachary N.
               and Lahijanian, Morteza},
  title     = {Automaton-guided control synthesis for signal temporal
               logic specifications},
  booktitle = {61st IEEE Conference on Decision and Control (CDC)},
  pages     = {3243--3249},
  year      = {2022}
}

@article{takayama2025stlccp,
  author  = {Takayama, Yoshinari and Hashimoto, Kazumune and Ohtsuka,
             Toshiyuki},
  title   = {{STLCCP}: Efficient convex optimization-based framework for
             signal temporal logic specifications},
  journal = {IEEE Transactions on Automatic Control},
  volume  = {70},
  number  = {9},
  pages   = {6064--6079},
  year    = {2025}
}

@misc{wehbeh2026logic,
  author        = {Wehbeh, Jad and Kerrigan, Eric C.},
  title         = {Exact continuous reformulations of logic constraints in
                   nonlinear optimization and optimal control problems},
  howpublished  = {arXiv:2601.03906},
  year          = {2026}
}

@misc{decarli2026stlgcs,
  author        = {De Carli, Nicola and Marchesini, Gregorio and
                   Dimarogonas, Dimos V.},
  title         = {{STL-GCS}: A planner-controller framework for signal
                   temporal logic via graphs of time-varying convex sets},
  howpublished  = {arXiv:2607.19196},
  year          = {2026}
}

@misc{ye2026zeroshot,
  author        = {Ye, Bowen and Hou, Ancheng and Huang, Junyue and Liu,
                   Ruijia and Yin, Xiang},
  title         = {Zero-shot signal temporal logic planning with
                   disjunctive branch selection in dynamic semantic maps},
  howpublished  = {arXiv:2605.01222},
  year          = {2026}
}

@article{gilpin2021,
  author  = {Gilpin, Yann and Kurtz, Vince and Lin, Hai},
  title   = {A smooth robustness measure of signal temporal logic for
             symbolic control},
  journal = {IEEE Control Systems Letters},
  volume  = {5},
  number  = {1},
  pages   = {241--246},
  year    = {2021}
}

@inproceedings{pant2017,
  author    = {Pant, Yash Vardhan and Abbas, Houssam and Mangharam, Rahul},
  title     = {Smooth operator: Control using the smooth robustness of
               temporal logic},
  booktitle = {IEEE Conference on Control Technology and Applications
               (CCTA)},
  pages     = {1235--1240},
  year      = {2017}
}

@misc{han2025exact,
  author        = {Han, Shaohang and Verhagen, Joris and Tumova, Jana},
  title         = {Exact smooth reformulations for trajectory
                   optimization under signal temporal logic
                   specifications},
  howpublished  = {arXiv:2511.07375},
  year          = {2025}
}

@book{wolsey2020,
  author    = {Wolsey, Laurence A.},
  title     = {Integer Programming},
  edition   = {2nd},
  publisher = {Wiley},
  year      = {2020}
}

@inproceedings{toussaint2015,
  author    = {Toussaint, Marc},
  title     = {Logic-geometric programming: An optimization-based
               approach to combined task and motion planning},
  booktitle = {International Joint Conference on Artificial Intelligence
               (IJCAI)},
  pages     = {1930--1936},
  year      = {2015}
}

@article{wachter2006,
  author  = {W{\"a}chter, Andreas and Biegler, Lorenz T.},
  title   = {On the implementation of an interior-point filter
             line-search algorithm for large-scale nonlinear
             programming},
  journal = {Mathematical Programming},
  volume  = {106},
  number  = {1},
  pages   = {25--57},
  year    = {2006}
}

@inproceedings{biere1999,
  author={Biere, Armin and Cimatti, Alessandro and Clarke, Edmund and Zhu, Yunshan},
  title={Symbolic model checking without {BDD}s},
  booktitle={Tools and Algorithms for the Construction and Analysis of Systems (TACAS)},
  series={LNCS}, volume={1579}, pages={193--207}, year={1999}
}

@inproceedings{kordabad2024dr,
  author={Bahari Kordabad, Arash and Vlahakis, Eleftherios E. and Lindemann, Lars
          and Dimarogonas, Dimos V. and Soudjani, Sadegh},
  title={Distributionally Robust Control for Chance-Constrained
         Signal Temporal Logic Specifications},
  booktitle={63rd IEEE Conference on Decision and Control (CDC)},
  year={2024},
  pages={1593--1598}
}

@inproceedings{kordabad2024cbf,
  author    = {Bahari Kordabad, Arash and Charitidou, Maria and
               Dimarogonas, Dimos V. and Soudjani, Sadegh},
  title     = {Control barrier functions for stochastic systems under
               signal temporal logic tasks},
  booktitle = {European Control Conference (ECC)},
  pages     = {3213--3219},
  year      = {2024}
}

@book{nemhauser1988,
  title={Integer and Combinatorial Optimization},
  author={Nemhauser, George L. and Wolsey, Laurence A.},
  year={1988},
  publisher={Wiley}
}

@techreport{hojny2025scip,
  author={C. Hojny and M. Besan{\c{c}}on and others},
  title={The {SCIP} Optimization Suite 10.0},
  institution={Optimization Online},
  year={2025},
  note={arXiv:2511.18580}
}

@article{burer2012non,
  title={Non-convex mixed-integer nonlinear programming: A survey},
  author={Burer, Samuel and Letchford, Adam N},
  journal={Surveys in Operations Research and Management Science},
  volume={17},
  number={2},
  pages={97--106},
  year={2012},
  publisher={Elsevier}
}

@article{kordabad2025ddr,
  author  = {Bahari Kordabad, Arash and Vlahakis, Eleftherios E. and
             Lindemann, Lars and Gros, Sebastien and Dimarogonas,
             Dimos V. and Soudjani, Sadegh},
  title   = {Data-driven distributionally robust control for interacting
             agents under logical constraints},
  journal = {arXiv:2503.09816},
  year    = {2025}
}
\end{document}